\documentclass[journal,twoside,web]{ieeecolor}
\usepackage{lcsys}
\usepackage{cite}
\usepackage{amsmath,amssymb,amsfonts,amstext}
\usepackage{graphicx}
\usepackage{subcaption}
\usepackage{textcomp}
  \def\BibTeX{{\rm B\kern-.05em{\sc i\kern-.025em b}\kern-.08em
     T\kern-.1667em\lower.7ex\hbox{E}\kern-.125emX}}
\usepackage{todonotes}

\usepackage{bbold}
\usepackage{bbm}
\usepackage{mathrsfs}
\usepackage{verbatim}
\usepackage{amsmath}
\usepackage{amsfonts}
\usepackage{bm}
\usepackage{caption}
\usepackage{float}

\usepackage{amsthm}
\usepackage{mathtools}
\usepackage{cite}
\usepackage{hyperref}

\allowdisplaybreaks

\usepackage{algorithm}
\usepackage{algpseudocode}

\theoremstyle{definition}
\newtheorem{assumption}{Assumption}

\newtheorem{remark}{Remark}
\newtheorem{problem}{Problem}

\theoremstyle{plain}
\newtheorem{proposition}{Proposition}

\DeclareMathOperator*{\minimize}{Minimize}

\makeatletter
\let\NAT@parse\undefined
\makeatother
\usepackage{hyperref}
\hypersetup{
  colorlinks=true,
    linkcolor= blue,
    allcolors=blue,
    citecolor = blue,
    filecolor=black,      
    urlcolor=blue,
    }

\title{Safe Merging in Mixed Traffic with Confidence}

\author{Heeseung Bang$^{1,2}$, \textit{Student Member, IEEE}, Aditya Dave$^{2}$, \textit{Member, IEEE},\\ Andreas A. Malikopoulos$^{2}$, \textit{Senior Member, IEEE}
    \thanks{This research was supported by NSF under Grants CNS-2149520 and CMMI-2219761.}
    \thanks{$^{1}$Heeseung Bang is with the Department of Mechanical Engineering, University of Delaware, Newark, DE 19716, USA.}
    \thanks{$^{2}$ Adita Dave and Andreas A. Malikopoulos are with the School of Civil and Environmental Engineering, Cornell University, Ithaca, NY 14850, USA. {\tt\small email: \{hb489,a.dave,amaliko\}@cornell.edu}}
}

\date{February 2022}
\begin{document}

\maketitle
\thispagestyle{empty}
\begin{abstract}
In this letter, we present an approach for learning human driving behavior, without relying on specific model structures or prior distributions, in a mixed-traffic environment where connected and automated vehicles (CAVs) coexist with human-driven vehicles (HDVs). We employ conformal prediction to obtain theoretical safety guarantees and use real-world traffic data to validate our approach. Then, we design a controller that ensures effective merging of CAVs with HDVs with safety guarantees. We provide numerical simulations to illustrate the efficacy of the control approach.
\end{abstract}


\section{Introduction}
Connected and automated vehicles (CAVs) promise to revolutionize transportation. With 100\% CAV penetration, we can use control algorithms \cite{malikopoulos2019ACC} to improve the efficiency \cite{Bang2023flowbased}, safety \cite{Malikopoulos2020}, flow \cite{chalaki2023minimally}, and equity \cite{Bang2023mem} within any transportation network. However, we expect the percentage of deployed CAVs to rise gradually over time. Therefore, there is an urgent need for control approaches that ensure the effective operation of CAVs among human-driven vehicles (HDVs). 
This is a challenging problem where the inherent uncertainties associated with human driving clash with the need for safety, efficiency, and comfort of any CAV in mixed traffic.
Furthermore, any control algorithm deployed in CAVs must allow real-time implementation in diverse situations.
Therefore, in this letter, we focus on the control problem in mixed traffic when a CAV seeks to merge from an off-ramp onto a highway among oncoming HDVs. 



In recent years, the problem of merging in mixed traffic has attracted considerable attention. Game-theoretic models have been proposed for the interactions between an HDV and an oncoming CAV \cite{schwarting2019social}. Typically, these models consider that the HDV's driving decisions can be interpreted using a social value orientation that trades off between selfish and altruistic behavior. Such models can yield control algorithms for CAVs that utilize theoretical results from leader-follower games \cite{liu2021cooperation} and potential games \cite{liu2023potential}. 
Alternatively, HDV behavior can be framed within the context of car-following models, which postulate that HDVs drive in a manner that follows their preceding vehicle within a driving lane. For example, Gibbs' car following model has been utilized in conjunction with Q-learning to facilitate lane changing of CAVs among HDVs \cite{guo2020merging}. Similarly, Newell's car-following model was utilized with Bayesian optimization to derive a control algorithm that facilitates merging in mixed traffic \cite{Le2023Stochastic}. A typical control approach within all model-based approaches is to learn unknown parameters within the structure of the HDV model and use them to predict the HDV's future trajectory. Due to the known model structure, these approaches provide uncertainty quantification for human driving and performance guarantees on the eventual control algorithm of the CAV.
However, while these control approaches work well in simulation, their guarantees may not hold in practice if actual human driving behavior differs from the implicitly assumed model.

As an alternative to considering specific models for HDVs, data-driven approaches attempt to derive a control strategy with minimal assumptions on human-driving behavior. 
Deep neural networks with long-short-term memory (LSTM) cells have been proposed as general approximators to learn HDV behavior in various driving scenarios \cite{altche2017lstm, park2018sequence}.
By leveraging time-series data, these networks can learn the interdependencies of vehicle motion in high-traffic situations \cite{deo2018convolutional}.
After training a network of sufficient complexity to predict HDV behavior, it is possible to facilitate effective control of CAVs during highway on-ramp merging with control algorithms such as model-predictive control \cite{Nishanth2023AISmerging} or reinforcement learning algorithms such as deep policy gradient \cite{el2021novel}. 
Reinforcement learning has also served as a model-free approach to computing optimal merging strategies for CAVs directly from data \cite{yan2021reinforcement, peng2021connected}. While such data-driven approaches can learn the complex interactions among HDVs and CAVs, they suffer from a lack of formal guarantees and a lack of interpretability. Thus, their solutions cannot be implemented with confidence within the physical realm.


In this letter, we overcome these challenges by developing a control framework that uses conformal predictions for HDVs to merge data-driven learning with the formal safety guarantees of model-based approaches.
Conformal predictions present a calibration methodology to generate set-valued predictions from any black-box model. These set valued predictions are equipped with formal guarantees on the predictive uncertainty \cite{angelopoulos2023conformal}. Conformal prediction has been applied to many prediction problems, including dynamic time-series data \cite{zaffran2022adaptive, stankeviciute2021conformal} and hidden Markov models \cite{nettasinghe2023extending} across an infinite time horizon. It has also been utilized for safe model-predictive control with obstacles in robotics \cite{lindemann2023safe}. Our primary contributions are (1) the development of a framework integrating conformal predictions into the CAV merging problem and (2) the design of a safe learning and control approach utilizing this framework. A salient feature of our framework is that it can enhance the reliability and safety of CAVs with general assumptions on HDV behavior and without assuming specific prior distributions (e.g., Gaussian).





The remainder of the letter proceeds as follows.
In Section \ref{sec:problem}, we explain our modeling framework of CAV merging on-ramp in mixed traffic scenario.
In Section \ref{sec:learning}, we present a learning approach to predict human-driving behavior and resulting traffic conditions.
We introduce conformal prediction and provide a method to incorporate it into the safe merging problem in Section \ref{sec:prediction}. We provide a control strategy for safe merging and demonstrate the effectiveness of the framework through numerical simulations in Section \ref{sec:control}. Finally, we draw concluding remarks in Section \ref{sec:conclusion}.

\section{Modeling Framework} \label{sec:problem}

We consider a mixed traffic scenario where a CAV seeks to merge into a highway among oncoming HDVs (see Fig. \ref{fig:control-zone}).
We define a \emph{control zone} marked by the dashed boundary in Fig. \ref{fig:control-zone}, where we seek to plan the trajectory of the CAV and compute corresponding control inputs.
Since CAV-CAV interactions are well understood \cite{Malikopoulos2020, chalaki2023minimally, Bang2023flowbased}
our focus is on CAV-HDV interactions within the control zone.
We consider longitudinal movements only and assume that CAVs can communicate with infrastructure to obtain position and speed information of other vehicles, but the dynamics of HDVs may be unknown to the CAV.
Within the control zone, the CAV's motion is modeled using double integrator dynamics, i.e.,
\begin{equation}
    \begin{aligned}
    p_\mathrm{c}(t+1) &= p_\mathrm{c}(t) + v_\mathrm{c}(t) \Delta t + \frac{1}{2} u_\mathrm{c}(t) \Delta t^2,\\
    v_\mathrm{c}(t+1) &= v_\mathrm{c}(t) + u_\mathrm{c}(t) \Delta t,
    \end{aligned}
\end{equation}
where $\Delta t\in\mathbb{R}$ is a sampling time, $p_\mathrm{c}(t)\in\mathbb{R}$, $v_\mathrm{c}(t)\in\mathbb{R}$, and $u_\mathrm{c}(t)\in\mathbb{R}$ are the position, speed, and control input at time step $t \in \mathcal{T}$, respectively.
The CAV's speed and acceleration limits are given by
\begin{align}
    u_{\text{c}}^{\text{min}} & \leq u_\mathrm{c}(t) \leq u_{\text{c}}^{\text{max}}, \label{eqn:ulim}\\
    0 < v^{\text{min}} & \leq v_\mathrm{c}(t) \leq v^{\text{max}}, \label{eqn:vlim}
\end{align}

Let $\mathcal{N}=\{1,\dots,N\}$ denote the set of HDVs on the highway.
For each HDV $n\in\mathcal{N}$, the variables $p_n(t)$ and $v_n(t)$ represent position and speed at time $t\in \mathcal{T} := \{0,\dots,T\}$, respectively. As time and position have a one-to-one mapping in highway-driving (non-stopping) vehicles, there exists a function $\tau_n$ that yields a certain time for HDV $n$ to arrive at a specific position, i.e., $\tau_n(\cdot) = p_n^{-1}(\cdot)$.
For a CAV to merge onto the highway, we consider a finite set of merging candidates $\mathcal{L}=\{1,\dots,L\}$, where each candidate point $\ell\in\mathcal{L}$ is located in index order within a merging lane with equal distances (see Fig. \ref{fig:control-zone}).
To guarantee safety between HDVs and a CAV at a merging candidate $\ell\in\mathcal{L}$, we impose the time headway constraint
\begin{equation}
    |T^\mathrm{m} - \tau_n\left(p^{\mathrm{m},\ell}\right)| > \delta,~~\forall n\in\mathcal{N}, \label{eqn:time_headway} 
\end{equation}
where $p^{\mathrm{m},\ell}$ is the location of merging candidate $\ell$ and $T^\mathrm{m} \in \mathcal{T}$ is the merging time of the CAV.

\begin{figure}
    \centering
    \includegraphics[width=0.9\linewidth]{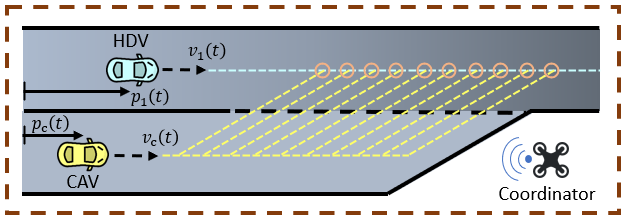}
    \caption{Control zone with one HDV on the highway. The dashed lines represent potential trajectories of the HDV and the CAV. The orange circles represent merging candidates.}
    \label{fig:control-zone}
    \vspace{-14pt}
\end{figure}

For any given merging candidate $\ell$, we consider that the CAV follows an unconstrained energy-optimal trajectory \cite{Malikopoulos2020}, i.e., its trajectory is described by
\begin{align} \label{eq:optimalTrajectory}
    u_\mathrm{c}(t) &= 6 a_\mathrm{c}t + 2 b_\mathrm{c}, \notag \\
    v_\mathrm{c}(t) &= 3 a_\mathrm{c} t^2 + 2 b_\mathrm{c} t + c_\mathrm{c}, \\
    p_\mathrm{c}(t) &= a_\mathrm{c} t^3 + b_\mathrm{c} t^2 + c_\mathrm{c} t + d_\mathrm{c}, \notag
\end{align}
where $a_\mathrm{c}, b_\mathrm{c}, c_\mathrm{c}$, and $d_\mathrm{c}$ are constants of integration, which can be computed using the following boundary conditions
\begin{align}
     p_\mathrm{c}(t_\mathrm{c}^0) = 0,&\quad  v_\mathrm{c}(t_\mathrm{c}^0)= v_\mathrm{c}^0 , \label{eqn:bci}\\
     p_\mathrm{c}(T^\mathrm{m}) = p^{\mathrm{m},\ell},&\quad v_\mathrm{c}(T^\mathrm{m})=v_\mathrm{c}^\mathrm{m}.\label{eqn:bcf}
\end{align}
Since $v^{\text{min}}>0$ from \eqref{eqn:vlim}, there exists a one-to-one mapping from the position trajectory to the merging time. Consequently, we obtain inverse function with respect to $\bm{\psi}=[a_\mathrm{c}, b_\mathrm{c}, c_\mathrm{c}, d_\mathrm{c}]$, i.e., $ T^\mathrm{m} = p_c^{-1}(\bm{\psi})$.
The CAV's optimization problem is formulated as follows.

\begin{problem} \label{pb:merging}
    For a CAV to safely merge in between HDVs, we solve the following optimization problem
    \begin{align}
        \minimize_{\bm{\psi}, v_\mathrm{c}^\mathrm{m}, \ell}~ & T^\mathrm{m}\\
        \emph{subject to: }& \eqref{eqn:ulim} - \eqref{eqn:bcf},~\ell\in\mathcal{L}\notag.
    \end{align}
\end{problem}

The solution to Problem \ref{pb:merging} would enable a CAV to safely merge onto highways while minimizing speed transitions, indirectly enhancing energy consumption and comfort.
However, solving this problem requires predictions on the time $\tau_n(\cdot)$ of each HDV $n \in \mathcal{N}$ to pass each merging candidate $\ell \in \mathcal{L}$. Such predictions are not known a priori and must be obtained from an HDV driving model. 
Several efforts have been made to learn human-driving models, but these approaches often impose strong assumptions about specific probability distributions or structures for the model \cite{schwarting2019social, liu2021cooperation,liu2023potential,guo2020merging,Le2023Stochastic}.
Next, we present a very general model for human driving inspired by Markov decision processes. In subsequent sections, we show how LSTM networks can be employed to learn human-driving behavior and conformal predictions can offer theoretical guarantees to the learned model.

\subsection{Human Driving Behavior}

We consider a decision-making process for HDVs as illustrated in Fig. \ref{fig:human_model}. The observation of each HDV $n \in \mathcal{N}$ at any time $t$ is $o_n(t)=(x_{n-1}(t),x_n(t),x_{n+1}(t),x_{\mathrm{c}}(t))$, where $n-1$ denotes the vehicle leading the HDV $n$, $n+1$ is the vehicle preceding the HDV $n$, and $x_{*}(t) = \big(p_{*}(t), v_{*}(t) \big)$ is the state of a vehicle at time $t$.
The driver of HDV $n$ also receives an instinctive driving impulse, modeled as an uncontrolled disturbance $w_n(t)$ that takes values in a set $\mathcal{W}$. This disturbance captures the inherent stochasticity of human driving, and we consider that $\{w_n(t): t \in \mathcal{T}\}$ is a sequence of independent variables. Finally, we denote the internal state of the driver of HDV $n$ at time $t$ by $s_n(t)\in\mathcal{S}$ and consider this internal state to be an abstract representation of the current driving intent of the human, including factors such as aggression, biases, conservativeness, attentiveness, etc. To keep our framework general, we do not assume knowledge of the space $\mathcal{S}$. However, we do assume that starting with $s_n(0)$, the internal state $s_n(t)$ evolves as a Markov chain for all $t$:
\begin{equation}
    s_n(t+1) = \phi\left(s_n(t),o_n(t+1),w_n(t)\right). \label{eqn:internal_state}
\end{equation}
Based on the internal state, the driver of HDV $n$ takes an action $u_n(t) = g\left(s_n(t)\right) \in \mathbb{R}$. Substituting this action into the vehicle dynamics, the state of HDV $n$ evolves as 
\begin{equation} \label{eqn:HDV_dynamics}
    x_n(t+1) = f\big(x_n(t),s_n(t)\big).
\end{equation}
The primitive variables that determine the motion of HDV $n$ through \eqref{eqn:internal_state} - \eqref{eqn:HDV_dynamics} are $(x_n(0), s_n(0), w_n(t): t \in \mathcal{T})$. We consider that a CAV receives the observation $o_n(t)$ of the HDV through the coordinator in the control zone. However, the CAV has no knowledge of the internal spaces $\mathcal{S}$ and $\mathcal{W}$, and receives no observations of $s_n(t)$ or $w_n(t)$ for all $n \in \mathcal{N}$ and all $t \in \mathcal{T}$. We impose the following assumptions in our formulation.



\begin{assumption} \label{asm:prior}
We consider that the joint prior distribution on the primitive variables of all HDVs is exchangeable when a CAV has no additional information.
\end{assumption}

Assumption \ref{asm:prior} allows us to treat any two primitives from $\{(x_n(0), s_n(0), w_n(t): t \in \mathcal{T})\}_{n \in \mathcal{N}}$ as exchangeable, a property that will be essential to using conformal prediction in Section \ref{sec:prediction}. This assumption implies that re-ordering the HDVs has no effect on the joint distribution,  
and that a human driver does not distinguish among CAVs and HDVs.
Note that, in real world scenarios, a truck may not be exchangeable with a passenger vehicle. However, if features such as vehicle types or driving purpose are available, they can be added to the list of primitives to ensure exchangeability.

\begin{remark}
    The joint prior on all variables is defined on a probability space $(\Omega, \mathcal{F},\mathbb{P})$, with a sample space $\Omega$, $\sigma$-algebra $\mathcal{F}$ and probability measure $\mathbb{P}$. 
\end{remark}

\begin{remark}
    Consider any two random variables $z^1:(\Omega,\mathcal{F}) \to (Z,\mathcal{Z})$ and $z^2:(\Omega,\mathcal{F}) \to ({Z},\mathcal{Z})$, where $\mathcal{Z}$ is a $\sigma$-algebra on set ${Z}$. Then, $z^1$ and $z^2$ are \textit{exchangeable} if for any $A, B \in \mathcal{Z}$, we have $\mathbb{P}(z^1 \in A, z^2 \in B) = \mathbb{P}(z^1 \in B, z^2 \in A)$.
\end{remark}

\begin{assumption} \label{asm:influence}
Human driving behavior is influenced by the behavior of surrounding vehicles, including the leading, rear, and merging vehicles.
\end{assumption}

In our formulation, we restrict attention to the surrounding vehicles when defining the observation of an HDV and do not consider other environmental or weather-related effects. Considering such factors within the human driving model presents an intriguing direction for future research.

\begin{figure}[t!]
    \centering
    \includegraphics[width=0.9\linewidth]{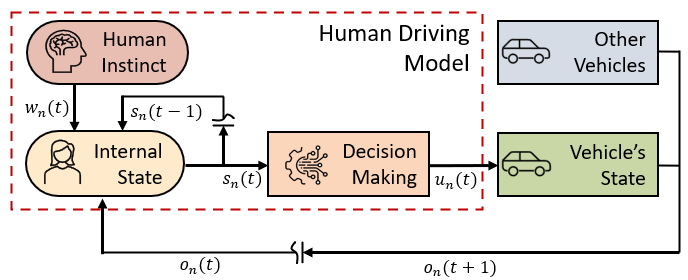}
    \caption{Decision-making loop for a human driver.}
    \label{fig:human_model}
    \vspace{-14pt}
\end{figure}

\section{Learning Human Driving Behavior} \label{sec:learning}

In this section, we present an approach to learn a data-driven model for human driving decisions. To this end, we use an encoder-decoder neural-network architecture with an LSTM, inspired by the notion of approximate information states for partially observed reinforcement learning \cite{subramanian2022approximate, Nishanth2023AISmerging}.
The encoder comprises two fully connected linear layers of dimensions $(8,10)$ and $(10,6)$ with ReLU activation, followed by an LSTM cell with a hidden state of size $6$ denoted by $\hat{s}_n(t) \in \mathbb{R}^6$ at any $t \in \mathcal{T}$.
The input to the linear layers at each $t$ includes the most recent observation $o_n(t)\in\mathbb{R}^8$ of an HDV $n \in \mathcal{N}$, which is supplied as a tuple of the surrounding vehicles' and its own states. The output of the linear layers is fed into the LSTM in addition to the previous hidden state $\hat{s}_n(t-1)$, and the LSTM's output is the hidden state $\hat{s}_t$ at time $t$. Note that as a result of the recursive nature of the LSTM, at each $t\in\mathcal{T}$, we can interpret the hidden state of the LSTM as a compression of the history of observations of HDV $n$, i.e., $\hat{s}_n(t) = \sigma(h_n(t),t)$ for some compression function $\sigma(\cdot)$, where $h_n(t)=(o_n(0),\dots,o_n(t))$ is the history.
Then, we can treat the hidden state of the LSTM $\hat{s}_t$ as a deterministic approximation of the human's internal state $s_t$ for all $t \in \mathcal{T}$ \cite{subramanian2022approximate, Nishanth2023AISmerging}. 
Subsequently, the output of the encoder $\hat{s}_n(t)$ is fed into the decoder to generate an estimation of the arrival time of HDV $n$ at each merging candidate $\ell$, i.e., $\mu^\ell\left(\hat{s}_n(t)\right)$.
The decoder consists of $L$ sequential modules with two linear layers each with a dimension of $(6,8)$ and $(8,1)$, each layer equipped with ReLU activation.


\subsection{Training} \label{subsection:training}

We train the network using real-world vehicle trajectories recorded at exits and entries of highways \cite{exiDdataset}. The dataset was collected with a sampling time of $0.04$ seconds.
At each sampling time, we capture the longitudinal positions and speeds of the vehicles on the merging lane and the highway lane adjacent to the merging lane, and we gather that information as a history of observation of each vehicle.
Our prediction target comprised the arrival times of these vehicles at all merging candidates ($L=10$).
We utilized the mean-square error as the loss function during training to minimize the discrepancy between our predictions and the target data.

Figure \ref{fig:training_result} provides an example of the prediction with our trained model, illustrating the predicted (relative) arrival time of an HDV to $10$ merging candidates. This visualization highlights the nature of our predictor: within a trajectory, uncertainty diminishes as the predictor accumulates more data.

\begin{figure}
    \centering
    \includegraphics[width=0.49\linewidth]{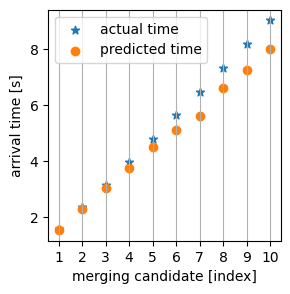}
    \includegraphics[width=0.49\linewidth]{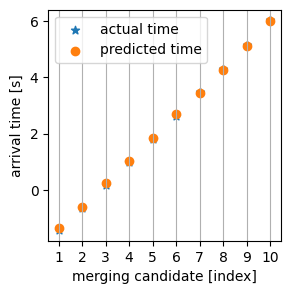}
    \caption{Prediction of HDV's arrival time to merging candidates with our trained model. The prediction were made at time $t=0$ (left) and $t=T^\mathrm{m}/3$ (right).}
    \label{fig:training_result}
   
\end{figure}





\section{Prediction with Confidence} \label{sec:prediction}

In this section, we employ conformal prediction on the trained HDV model to obtain a set-valued predictive region for the arrival time of any HDV $n \in \mathcal{N}$ at a merging candidate $\ell \in \mathcal{L}$.
Conformal prediction is a method for establishing such a predictive region for complex or black-box predictive models, such as our encoder-decoder architecture. It provides theoretical guarantees on the uncertainty associated with the predicted range of outputs in the form of a confidence measure without making assumptions on the form of the underlying distribution of the inputs or outputs (e.g., Gaussian).
As a consequence, conformal prediction has the potential to enhance the utilization of a trained model in safety-critical applications. 

Consider the calibration dataset with $K \in \mathbb{N}$ data points $\mathcal{D} = \{(z_k,y_k)\}_{k=1}^K$, where $z_k$ denotes the input and $y_k$ represents the corresponding output. For the next potential data point $(z_{K+1},y_{K+1})$ that we may receive, conformal prediction seeks a set $\mathcal{C}$ such that $(z_{K+1},y_{K+1}) \in \mathcal{C}$ with a high probability, i.e.,
\begin{equation}
\mathbb{P} \left( (z_{K+1},y_{K+1}) \in \mathcal{C} \right) \geq 1-\varepsilon,
\end{equation}
where $1-\varepsilon \in (0,1)$ is the confidence in $\mathcal{C}$. Furthermore, conformal prediction demands that the set $\mathcal{C}$ be tighter for predictions likely to be accurate and larger for predictions likely to be inaccurate. Thus, conformal predictions quantify the uncertainty associated with a prediction.

An essential condition for employing conformal prediction to construct a valid set $\mathcal{C}$ in the above example is that the data should be \textit{exchangeable}. This assumption is naturally satisfied for static prediction problems such as categorization and labeling.
However, its usage on the time series data has been limited due to the inherent non-exchangeability and interdependencies of data points across the trajectory of a time series.
Recent research efforts have attempted to overcome this limitation
by utilizing adaptive conformal prediction, where the guarantees are asymptotic and may not hold for a finite time horizon \cite{zaffran2022adaptive}. Next, we present an alternate approach to employ conformal prediction specialized to the CAV-HDV merging problem that does not require adaptive prediction and provides tight guarantees.

\subsection{Conformal Prediction for Vehicle Merging}

In our formulation, since the merging of vehicles occurs within a finite time frame, we can construct a conformal predictor for each $t \in\mathcal{T}$. To effectively deploy multiple conformal predictors, we only require two distinct trajectories to be exchangeable and do not need the data to be exchangeable across time within a single trajectory. We will prove in this subsection that exchangeability holds for two distinct trajectories when HDVs follow the general model in Fig. \ref{fig:human_model}. 

To begin creating the predictive ranges, we require a calibration dataset distinct from the training data used in Subsection \ref{subsection:training}. 
We denote this calibration dataset comprising of $K \in \mathbb{N}$ HDV trajectories as $\mathcal{D}_{h} := \{(h_{k}(t),\bm{\tau}_k: t \in \mathcal{T})\}_{k=1}^K$, where $h_k(t)$ represents the history of observations of the $k$-th HDV in the dataset at time $t$, and $\bm{\tau}_k = (\tau_k^1, \dots, \tau_k^L)$ is the vector of actual arrival times of the $k$-th HDV at all merging candidates during the trajectory. Note that within $\mathcal{D}_h$, the $k-1$-th HDV is not necessarily the leading vehicle of the $k$-th HDV, as these trajectories are selected randomly from a large collection of trajectories.
Next, we prove that under Assumption \ref{asm:prior}, the arrival times for all HDVs in $\mathcal{D}_h$ are exchangeable.



\begin{proposition} \label{prop:exchangeable}
Let $\hat{s}_k(t) = \sigma(h_k(t), t)$ be the hidden state of an LSTM from Section \ref{sec:learning} for any $h_k(t)$ in the dataset $\mathcal{D}_h$, at any $t\in\mathcal{T}$. Then, in the corresponding dataset $\mathcal{D}_{\hat{s}} := \{(\hat{s}_k(t), \bm{\tau}_k: t \in \mathcal{T})\}_{k=1}^K$, the tuples $(\hat{s}_i(t), \tau_i^{\ell})$ and $(\hat{s}_j(t), \tau_j^{\ell})$ are exchangeable for all $i,j \in \{1,\dots,K\}$, $\ell \in \mathcal{L}$, and $t \in \mathcal{T}$.
\end{proposition}
\begin{proof}
    Assumption \ref{asm:prior} implies that the primitives of any two HDV's are exchangeable. Using this assumption with  \eqref{eqn:internal_state} - \eqref{eqn:HDV_dynamics}, note that the states $x_i(t)$ and $x_j(t)$ of any two HDVs $i$ and $j$ are exchangeable for all $t \in \mathcal{T}$. Furthermore, since the history of each HDV is a collection of states and the arrival time is a deterministic function of the states and primitives, this also implies that $(h_i(t), \tau_i^{\ell})$ and $(h_j(t), \tau_j^{\ell})$ are exchangeable for all $\ell \in \mathcal{L}$ and $t \in \mathcal{T}$.
    Next, consider Borel subsets $S^{1}, S^{2} \subset \mathbb{R}^{6}$ and $T^{\ell,1}, T^{\ell,2} \in \mathbb{R}_{\geq0}$ for the realizations of the LSTM's hidden state and HDV's arrival time at candidate $\ell$, respectively. 
    Let $\mathbb{P}({S}^{1}, {S}^{2}, {T}^{\ell,1}, T^{\ell,2}) := \mathbb{P}\big(\hat{s}_i(t) \in S^{1}, \, \hat{s}_j(t) \in S^{2}, \, \tau_i^{\ell} \in T^{\ell,1}, \, \tau_j^{\ell} \in T^{\ell,2}\big)$ and $\mathbb{P}({S}^{2}, {S}^{1}, {T}^{\ell,2}, T^{\ell,1}) := \mathbb{P}\big(\hat{s}_i(t) \in S^{2}, \, \hat{s}_j(t) \in S^{1}, \, \tau_i^{\ell} \in T^{\ell,2}, \, \tau_j^{\ell} \in T^{\ell,1}\big)$. Then, using probability density functions to simplify the notation, $\mathbb{P} \left(S^1,S^2,T^1,T^2 \right) = \int_{S^{1,2}} \int_{T^{1,2}} \int_{H^{1,2}} p \left( s^1,s^2,\tau^1,\tau^2,h^1,h^2 \right) ds~d\tau~dh =\int_{S^{1,2}} \int_{T^{1,2}} \int_{H^{1,2}} \mathbb{I}\left[ s^1 = \sigma(h^1)\right] \cdot \mathbb{I}\left[ s^2 = \sigma(h^2)\right] \cdot p \left( \tau^1, \tau^2, h^1, h^2 \right) ds~d\tau~dh =\int_{S^{1,2}} \int_{T^{1,2}} \int_{H^{1,2}} \mathbb{I}\left[ s^2 = \sigma(h^2)\right]\cdot \mathbb{I}\left[ s^1 = \sigma(h^1)\right]\cdot p \left( \tau^2, \tau^1, h^2, h^1 \right) ds~d\tau~dh = \mathbb{P} \left(S^2,S^1,T^2,T^1 \right)$,
    where $\mathbb{I}$ is the indicator function, $s = (s^1, s^2)$, $\tau = (\tau^{\ell,1}, \tau^{\ell,2})$ and $h = (h^1, h^2)$, and $S^{1,2} = S^1 \times S^2$, $T^{1,2} = T^{\ell,1} \times T^{\ell,2}$ and $H^{1,2}$ is the product of the spaces of the histories. Here, in the first equality, we use the definition of the LSTM's internal state as a compression of the history, and in the second equality, we use the exchangeability of $(h_i(t), \tau_i^{\ell})$ and $(h_j(t), \tau_j^{\ell})$.
\end{proof}

Proposition \ref{prop:exchangeable} enables the use of conformal prediction directly for the input-output pair $(\hat{s}_k(t), \tau_k^{\ell})$ for any $t \in \mathcal{T}$ and merging candidate $\ell \in \mathcal{L}$. 
Given dataset $\mathcal{D}_{\hat{s}}$, for each $\ell\in\mathcal{L}$ and confidence level $1-\varepsilon$, we seek to find a bound $C^\ell(t)\in\mathbb{R}_{\geq0}$ for the $(K+1)$-the trajectory, such that
\begin{equation} \label{eqn:cp_prob}
\mathbb{P} \big( \big|\tau_{K+1}^\ell - \mu^\ell(\hat{s}_{K+1}(t))\big| < C^\ell(t) \big) \geq 1-\varepsilon,
\end{equation}
where, recall that $\mu^{\ell}(\cdot)$ refers to the decoder of the neural-network trained in Subsection \ref{subsection:training} and its output $\mu^{\ell}(\hat{s}_{K+1}(t))$ is the CAV's prediction at any time $t \in \mathcal{T}$ for the arrival time of the $(K+1)$-th HDV at candidate $\ell$. We construct $C^\ell(t)$ using a standard technique for regression problems \cite{lindemann2023safe}.
Let $r_k^{\ell}(t):=|\tau_k^\ell - \mu^\ell(\hat{s}_k(t))|$ be a nonconformity score, which measures the inaccuracy of the prediction.
For each HDV $n$ and merging candidate $\ell$, we compute all the nonconformity scores from the given dataset $\mathcal{D}_{\hat{s}}$ at each $t \in \mathcal{T}$. Then, we sort them into ascending order, i.e., we consider the re-ordering of scores $\bar{r}^{\ell}_{1}(t),\dots,\bar{r}^{\ell}_{K}(t)$ such that $\bar{r}^{\ell}_{i}(t) \geq \bar{r}^{\ell}_{i-1}(t)$ for all $i = 2, \dots, K$.
Then, we determine the conformal bound as
\begin{equation} \label{eqn:conformal_bound}
    C^{\ell}(t) = \bar{r}^{\ell}_q(t)~~\text{with}~~q:=\lceil (K+1)(1-\varepsilon) \rceil,
\end{equation}
where $\lceil \cdot \rceil$ is the ceiling function. 

\begin{remark} \label{remark:conformal_range}
    Since Proposition \ref{prop:exchangeable} establishes the exchangeability of all HDVs, during system operation, the CAV can treat any HDV $n \in \mathcal{N}$ generating the $(K+1)$-th trajectory. Thus, during operation, the conformal range constructed by the CAV at time $t$ for the arrival time of HDV $n \in \mathcal{N}$ at candidate $\ell$ is 
    \begin{equation}
        \mathcal{C}_{n}^{\ell}(t) = \left[\mu^{\ell}(\hat{s}_n(t)) - C^{\ell}(t) \,,\, \mu^{\ell}(\hat{s}_n(t)) + C^{\ell}(t)\right].
    \end{equation}
    For any decoder $\mu^{\ell}(\cdot)$, this range contains the true arrival time $\tau_n(p^{\mathrm{m},\ell})$ of HDV $n \in \mathcal{N}$ with probability of $1-\varepsilon$ \cite[Theorem 1]{tibshirani2019conformal}.
\end{remark}


Using \eqref{eqn:conformal_bound}, we constructed conformal ranges for the predictions of the neural network trained in Section \ref{sec:learning} using $500$ calibration trajectories (distinct from the training data) and setting a confidence level of $(1-\varepsilon) = 0.9$. On validating the performance of the trained model equipped with a conformal predictor over $100$ trajectories distinct from both the training and calibration data sets, we observed that across time steps and different merging candidates, the true merging time of an HDV lay within the conformal range $91.28\%$ of the time. One instance of this testing is illustrated in Fig. \ref{fig:cp_result}, where blue solid lines indicate predictions of the HDV's arrival time at two different merging candidates across $t \in \mathcal{T}$, and the shaded area indicates conformal range with confidence level $0.9$. Note that the conformal range shrinks over time. This is because the prediction uncertainty of the neural network decreases with time as the CAV receives more observations of the HDV's behavior and better anticipates future actions of the HDV.

\begin{remark}
    Note that an oncoming CAV merges only at one candidate from $\mathcal{L}$. Thus, to guarantee safe merging with the confidence of $1-\varepsilon$, it is sufficient to construct a conformal range for each candidate $\ell \in \mathcal{L}$ independent of the other candidates.
    As our approach does not seek guarantees on the prediction of all merging candidates simultaneously, it improves the tightness of our conformal ranges when compared with similar offline approaches that predict complete trajectories, e.g., \cite{lindemann2023safe}.
\end{remark}

\begin{figure}
    \centering
    \includegraphics[width=0.8\linewidth]{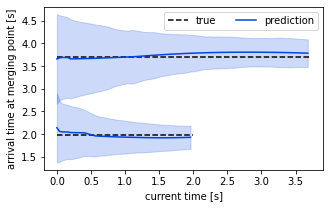}
    \caption{Prediction and conformal prediction range for two merging candidates at each time step.}
    \label{fig:cp_result}
    \vspace{-14pt}
\end{figure}

\section{Safe Merging Strategy} \label{sec:control}

In this section, we introduce a control framework for CAV merging that incorporates the conformal prediction bound $C^\ell(t)$.
Recall that in constraint \eqref{eqn:time_headway}, Problem \ref{pb:merging} requires the actual arrival time $\tau_n(p^{\mathrm{m},\ell})$ from the time trajectory of each HDV $n \in \mathcal{N}$ for each $\ell \in \mathcal{L}$, which is unknown to the CAV. Thus, at each $t \in \mathcal{T}$, the CAV utilizes a model from Subsection \ref{sec:learning} to predict the arrival time $\mu^{\ell}(\hat{s}_n(t))$ with a conformal range $\mathcal{C}^{\ell}_n(t)$ from Remark \ref{remark:conformal_range}. Then, the CAV aims to satisfy the following surrogate constraint at each $t \in \mathcal{T}$ instead of \eqref{eqn:time_headway}:
\begin{equation}
    |\mu^\ell(\hat{s}_n(t))-T^\mathrm{m}| \geq \delta + C^\ell(t).\label{eqn:time_headway_cp}
\end{equation}
Note that imposing \eqref{eqn:time_headway_cp} for all $n \in \mathcal{N}$ ensures that $\mathbb{P}\big(|T^\mathrm{m} - \tau_n\left(p^{\mathrm{m},\ell}\right)| > \delta,~\forall n\in\mathcal{N}\big) \geq 1 - \varepsilon$ during the planning at any $t \in \mathcal{T}$. Next, we present the CAV's new merging problem.

\begin{problem} \label{pb:merging_with_cp}
    At each $t \in \mathcal{T}$, the CAV solves the following optimization problem for a given confidence level $1-\varepsilon$
    \begin{align}
        \minimize_{\bm{\psi}, v_\mathrm{c}^\mathrm{m}, \ell}~ & T^\mathrm{m}\\
        \emph{subject to: }& \eqref{eqn:ulim}, \eqref{eqn:vlim}, \eqref{eq:optimalTrajectory} - \eqref{eqn:bcf}, \eqref{eqn:time_headway_cp}, \ell\in\mathcal{L}. \notag
    \end{align} 
\end{problem}

In Problem \ref{pb:merging_with_cp}, the CAV's goal is to minimize the merging time $T^\mathrm{m}$ with confidence level $1-\varepsilon$. Note that a feasible solution may not exist in Problem \ref{pb:merging_with_cp} for each $t \in \mathcal{T}$. Next, we provide a simple condition to ensure recursive feasibility.

\begin{proposition} \label{prp:recursive}
    Consider Problem \ref{pb:merging_with_cp} has a feasible solution at $t=0$. Then, if $C^\ell(t+1) \leq C^\ell(t)$ for all $t\in\mathcal{T}$, Problem \ref{pb:merging_with_cp} has a feasible solution for all $t$.
\end{proposition}

\begin{proof}
For any $t \in \mathcal{T}$, consider a feasible $T^\mathrm{m}$ satisfying \eqref{eqn:time_headway_cp}. Then,
$|\mu^\ell(\hat{s}_n(t)) - T^\mathrm{m}| \geq \delta + C^\ell(t) \geq \delta + C^\ell(t+1)$, implying that $T^\mathrm{m}$ satisfies \eqref{eqn:time_headway_cp} at time $t+1$. Thus, the result holds by mathematical induction starting at $t=0$.
\end{proof}

Generally, the $C^\ell(t+1) \leq C^\ell(t)$ may not hold unless HDVs behave consistently across time. However, if an HDV does behave consistently, this assumption holds for a good predictor that systematically reduces uncertainty with increasing observations, as illustrated in Fig. \ref{fig:cp_result}. If no new HDVs enter the control zone beyond those in $\mathcal{N}$, recursive feasibility holds under significantly more relaxed conditions because the CAV can slow down to let all HDVs pass before merging.


\subsection{Numerical Simulations}
\begin{figure}
    \centering
    \includegraphics[width=0.49\linewidth]{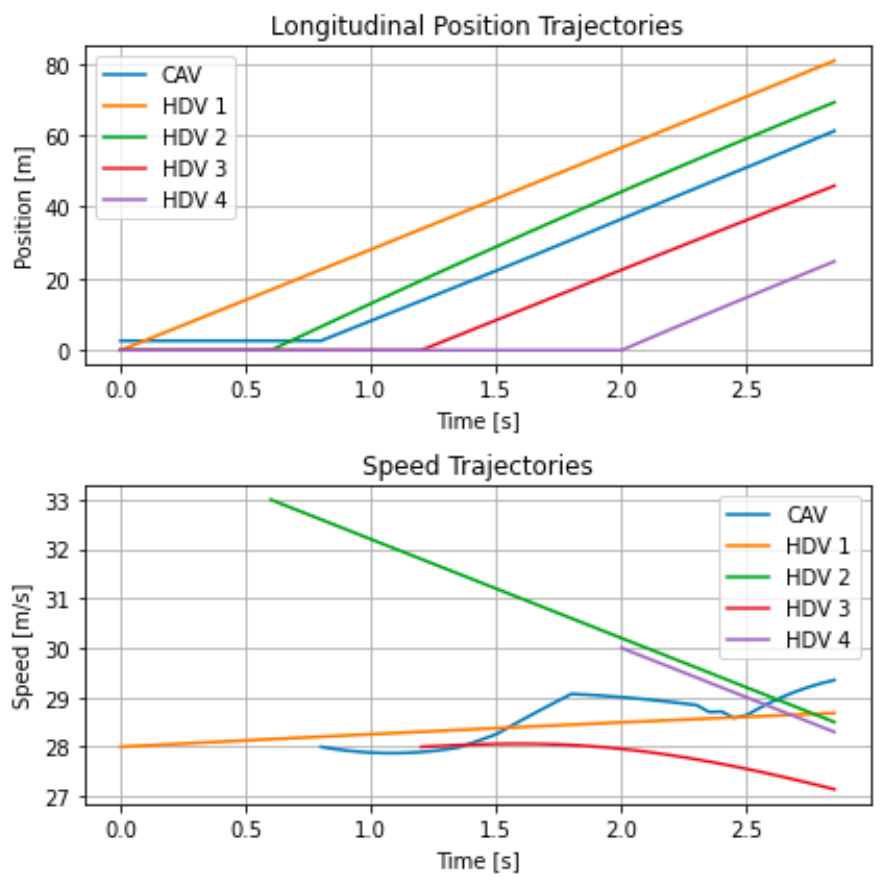}
    \includegraphics[width=0.49\linewidth]{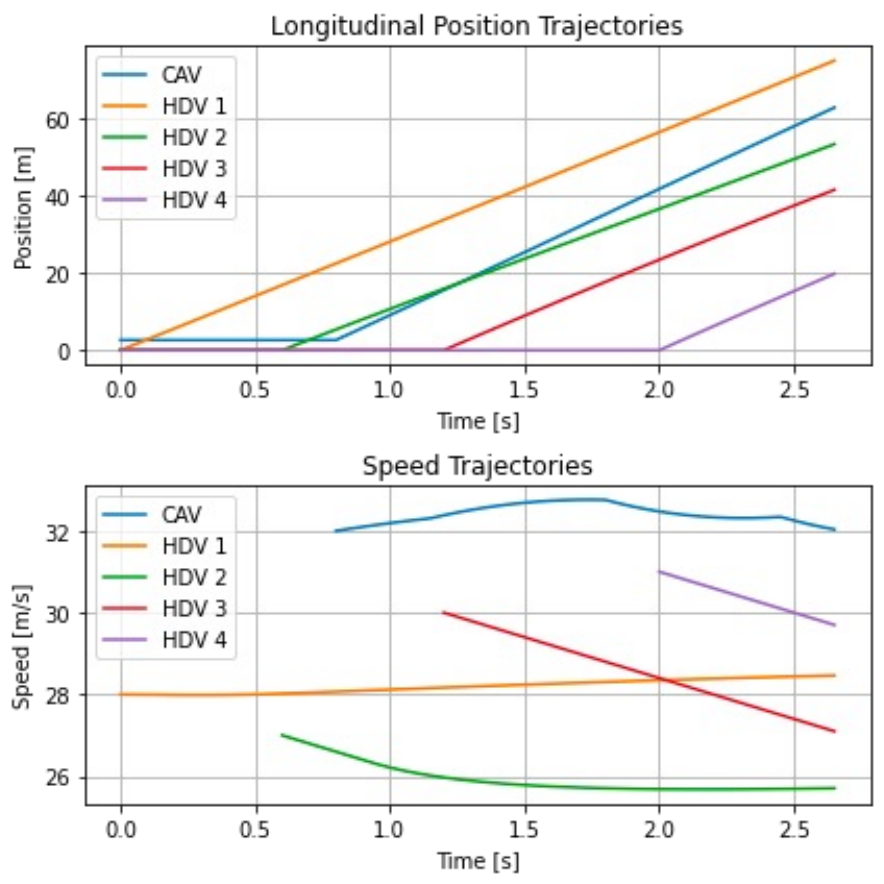}
    \caption{Simulation results for two different initial conditions.}
    \label{fig:sim_results}
    \vspace{-14pt}
\end{figure}
To mimic real HDVs in the simulator, we adopt an intelligent driver model (IDM) with some modifications.
The acceleration of HDV $n$ is given by $u_n(t) = u_\mathrm{IDM}(t) - \rho_n \cdot \exp(-\alpha \cdot \Delta p(t)^2)$, where $u_\mathrm{IDM}(t)$ is an acceleration determined from the general IDM, $\rho_n\in\mathbb{R}_{\geq 0}$ is altruism level of HDV $n$, $\alpha \in \mathbb{R}_{\geq 0}$ is the driver's sensitivity to the CAV, and $\Delta p(t)$ is the longitudinal distance between HDV $n$ and the CAV.
An HDV $n$ with more altruism tends to yield to the CAV as the longitudinal distance decreases. Less altruistic HDVs ignore the CAV and follow the regular IDM. 

We used the neural network in Section \ref{sec:learning} to predict the behavior of each HDV. We trained the network and calibrated the predictions using $3000$ trajectories for each task. Then, we deployed the network to solve Problem \ref{pb:merging_with_cp} numerically.
Figure \ref{fig:sim_results} illustrates two different simulation results, where each column varies in initial speeds and altruism levels for HDVs and maintains the same initial conditions otherwise.
In the left column of Fig. \ref{fig:sim_results}, the speed profile indicates that the CAV attempted to merge initially between HDVs $1$ and $2$ but eventually merged behind HDV $2$. In contrast, the right column shows the CAV merging in front of HDV $2$. These distinct behaviors arise from differences in altruism levels among HDVs $2$ and $3$. These results demonstrate that our approach allows the CAV to adapt safely to various HDV behaviors.

\vspace{-10pt}

\section{Concluding Remarks} \label{sec:conclusion}

In this letter, we proposed a method of learning human driving behavior without assuming specific structure of prior distributions and employing conformal prediction to obtain theoretical guarantees on the learned model.
Utilizing our model, we presented a control framework for a CAV to merge safely in between HDVs.
We validated our approach using real-world traffic data and through numerical simulations.
An interesting avenue for future work is to dynamically optimize confidence levels and consider more complicated interactions.

\vspace{-6pt}

\bibliographystyle{IEEEtran}
\bibliography{Bang, IDS}

\end{document}